\documentclass[conference]{IEEEtran}

\makeatletter
\def\ps@headings{%
\def\@oddhead{}%
\def\@evenhead{}%
\def\@oddfoot{\mbox{}\scriptsize\rightmark \hfil \thepage}%
\def\@evenfoot{\scriptsize\thepage \hfil \leftmark\mbox{}}}
\makeatother

\usepackage[dvips]{color}
\usepackage{epsf}
\usepackage{times}
\usepackage{epsfig}
\usepackage{graphicx}
\usepackage{amsmath}
\usepackage{amssymb}
\usepackage{amsxtra}
\usepackage{here}
\usepackage{rawfonts}
\usepackage{times}
\usepackage{url}
\usepackage{cite}

\newtheorem{convention}{\bf Convention}

\newtheorem{theorem}{\bf Theorem}
\newtheorem{property}{\bf Property}

\newtheorem{lemma}{\bf Lemma}

\newtheorem{definition}{\bf Definition}

\IEEEoverridecommandlockouts
\begin{document}
\title{Coalitional Games for Distributed Collaborative Spectrum Sensing in Cognitive Radio Networks \vspace{-0.4cm} }

\author{\authorblockN{ Walid Saad$^1$, Zhu Han$^2$, M{\'e}rouane Debbah$^3$, Are Hj{\o}rungnes$^1$ and Tamer Ba\c{s}ar$^4$} \authorblockA{\small
$^1$UNIK - University Graduate Center, University of Oslo, Kjeller, Norway, Email: \url{{saad, arehj}@unik.no}\\
$^2$ Electrical and Computer Engineering Department, University of Houston, Houston, USA, Email: \url{zhan2@mail.uh.edu}\\
$^3$ Alcatel-Lucent Chair in Flexible Radio,  SUP{\'E}LEC, Gif-sur-Yvette, France, Email: \url{merouane.debbah@supelec.fr}\\
$^4$ Coordinated Science Laboratory, University of Illinois at Urbana-Champaign, USA, Email: \url{tbasar@control.csl.uiuc.edu}\vspace{-0.95cm}
 }%
   \thanks{This work was done during the stay of Walid Saad at the Coordinated Science Laboratory at the University Of Illinois at Urbana - Champaign and is supported by the Research Council of Norway
    through the projects 183311/S10, 176773/S10 and 18778/V11.}}
\date{}
\maketitle

\begin{abstract}
Collaborative spectrum sensing among secondary users~(SUs) in cognitive networks is shown to yield a significant performance improvement. However, there exists an inherent trade off between the gains in terms of probability of detection of the primary user (PU)  and the costs in terms of false alarm probability. In this paper, we study the impact of this trade off on the topology and the dynamics of a network of SUs seeking to reduce the interference on the PU through collaborative sensing. Moreover, while existing literature mainly focused on centralized solutions for collaborative sensing, we propose distributed collaboration strategies through game theory. We model the problem as a non-transferable coalitional game, and  propose a distributed algorithm for coalition formation through simple merge and split rules. Through the proposed algorithm, SUs can autonomously collaborate and self-organize into disjoint independent coalitions, while maximizing their detection probability taking into account the cooperation costs (in terms of false alarm). We study the stability of the resulting network structure, and show that a maximum number of SUs per formed coalition exists for the proposed utility model. 
Simulation results show that the proposed algorithm allows a reduction of up to $86.6\%$ of the average missing probability per SU (probability of missing the detection of the PU) relative to the non-cooperative case, while maintaining a certain false alarm level. In addition, through simulations, we compare the performance of the proposed distributed solution with respect to an optimal centralized solution that minimizes the average missing probability per SU.  Finally, the results also show how the proposed algorithm autonomously adapts the network topology to environmental changes such as mobility.
\end{abstract}

\section{Introduction}
In recent years, there has been an increasing growth in wireless services, yielding a huge demand on the radio spectrum. However, the spectrum resources are scarce and most of them have been already licensed to existing operators. Recent studies showed that the actual licensed spectrum remains unoccupied for large periods of time \cite{FCC}. In order to efficiently exploit these spectrum holes, \emph{cognitive radio} (CR) has been proposed \cite{CR00}. By monitoring and adapting to the environment, CRs (secondary users) can share the spectrum with the licensed users (primary users), operating whenever the primary user (PU) is not using the spectrum. Implementing such flexible CRs faces several challenges \cite{CR02}. For instance, CRs must constantly sense the spectrum in order to detect the presence of the PU and use the spectrum holes without causing harmful interference to the PU. Hence, efficient spectrum sensing constitutes a major challenge in cognitive networks.

For sensing the presence of the PU, the secondary users~(SUs) must be able to detect the signal of the PU. Various kinds of detectors can be used for spectrum sensing such as matched filter detectors, energy detectors, cyclostationary detectors or wavelet detectors \cite{DT00}. However, the performance of spectrum sensing is significantly affected by the degradation of the PU signal due to path loss or shadowing (hidden terminal). It has been shown that, through collaboration among SUs, the effects of this hidden terminal problem can be reduced and the probability of detecting the PU can be improved \cite{CS00,CS01,CS03}. For instance, in \cite{CS00} the SUs collaborate by sharing their sensing decisions through a centralized fusion center in the network. This centralized entity combines the sensing bits from the SUs using the OR-rule for data fusion and makes a final PU detection decision. A similar centralized approach is used in \cite{CS01} using different decision-combining methods. 
 The authors in \cite{CS03} propose spatial diversity techniques for improving the performance of collaborative spectrum sensing by combatting the error probability due to fading on the reporting channel between the SUs and the central fusion center. Existing literature mainly focused on the performance assessment of collaborative spectrum sensing in the presence of a centralized fusion center. However, in practice, the SUs may belong to different service providers and they need to interact with each other for collaboration without relying on a centralized fusion center. Moreover, a centralized approach can lead to a significant overhead and increased complexity.

The main contribution of this paper is to devise distributed collaboration strategies for SUs in a cognitive network. Another major contribution is to study the impact on the network topology of the inherent trade off that exists between the collaborative spectrum sensing gains in terms of detection probability and the cooperation costs in terms of false alarm probability. This trade off can be pictured as a trade off between reducing the interference on the PU (increasing the detection probability) while maintaining a good spectrum utilization (reducing the false alarm probability). For distributed collaboration, we model the problem as a non-transferable coalitional game and we propose a distributed algorithm for coalition formation based on simple merge and split rules. Through the proposed algorithm, each SU autonomously decides to form or break a coalition for maximizing its utility in terms of detection probability while accounting for a false alarm cost. We show that, due to the cost for cooperation, independent disjoint coalitions will form in the network. We study the stability of the resulting coalition structure and show that a maximum coalition size exists for the proposed utility model. Through simulations, we assess the performance of the proposed algorithm relative to the non-cooperative case, we compare it with a centralized solution and we show how the proposed algorithm autonomously adapts the network topology to environmental changes such as mobility.

The rest of this paper is organized as follows: Section~\ref{sec:prob} presents
the system model. In Section \ref{sec:coal}, we present the proposed coalitional game and prove different properties while in Section~\ref{sec:coalform} we devise a distributed algorithm for coalition formation. Simulation results are presented and analyzed in Section \ref{sec:sim}. Finally, conclusions are drawn in
Section \ref{sec:conc}.

\section{System Model}\label{sec:prob}
Consider a cognitive network consisting of $N$ transmit-receive pairs of SUs and a single PU. The SUs and the PU can be either stationary and mobile. Since the focus is on spectrum sensing, we are only interested in the transmitter part of each of the $N$ SUs. The set of all SUs is denoted $\mathcal{N} = \{1,\ldots,N\}$. In a non-cooperative approach, each of the $N$ SUs continuously senses the spectrum in order to detect the presence of the PU. For detecting the PU, we use energy detectors which are one of the main practical signal detectors in cognitive radio networks \cite{CS00,CS01,CS03}. In such a non-cooperative setting, assuming Rayleigh fading, the detection probability and the false alarm probability of a SU $i$ are, respectively, given by $P_{d,i}$ and $P_{f,i}$ \cite[Eqs. (2), (4)]{CS00}
\begin{align}\label{eq:detind}
P_{d,i} = e^{-\frac{\lambda}{2}} \sum_{n=0}^{m-2}\frac{1}{n!}\left(\frac{\lambda}{2}\right)^n + \left(\frac{1+\bar{\gamma}_{i,PU}}{\bar{\gamma}_{i,PU}}\right)^{m-1}
\nonumber\\
\times \left[e^{-\frac{\lambda}{2\left(1+\bar{\gamma}_{i,PU}\right)}} - e^{-\frac{\lambda}{2}} \sum_{n=0}^{m-2}\frac{1}{n!}\left(\frac{\lambda\bar{\gamma}_{i,PU}}{2(1+\bar{\gamma}_{i,PU})}\right)^n \right],
\end{align}

\begin{equation}\label{eq:falseind}
P_{f,i} = P_f= \frac{\Gamma(m,\frac{\lambda}{2})}{\Gamma(m)},
\end{equation}
where $m$ is the time bandwidth product, $\lambda$ is the energy detection threshold assumed the same for all SUs without loss of generality as in \cite{CS00,CS01,CS03}, $\Gamma(\cdot,\cdot)$ is the incomplete gamma function and $\Gamma(\cdot)$ is the gamma function. Moreover, $\bar{\gamma}_i$ represents the average SNR of the received signal from the PU to SU given by $\bar{\gamma}_{i,\textrm{PU}} = \frac{P_{\textrm{PU}}h_{\textrm{PU},i}}{\sigma^2}$ with $P_{\textrm{PU}}$ the transmit power of the PU,  $\sigma^2$ the Gaussian noise variance and $h_{\textrm{PU},i}=\kappa/d_{\textrm{PU},i}^{\mu}$ the path loss between the PU and SU $i$; $\kappa$ being the path loss constant, $\mu$ the path loss exponent and $d_{\textrm{PU},i}$ the distance between the PU and SU $i$. It is important to note that the non-cooperative false alarm probability expression depends \emph{solely} on the detection threshold $\lambda$ and does not depend on the SU's location; hence we dropped the subscript $i$ in (\ref{eq:falseind}).

Moreover, an important metric that we will thoroughly use is the missing probability for a SU $i$, which is defined as the probability of missing the detection of a PU and given by \cite{CS00}
\begin{equation}\label{eq:missprob}
P_{m,i} = 1 - P_{d,i}.
\end{equation}
For instance, reducing the missing probability directly maps to increasing the probability of detection and, thus, reducing the interference on the PU.
\begin{figure}[t]
\begin{center}
\includegraphics[scale=0.5]{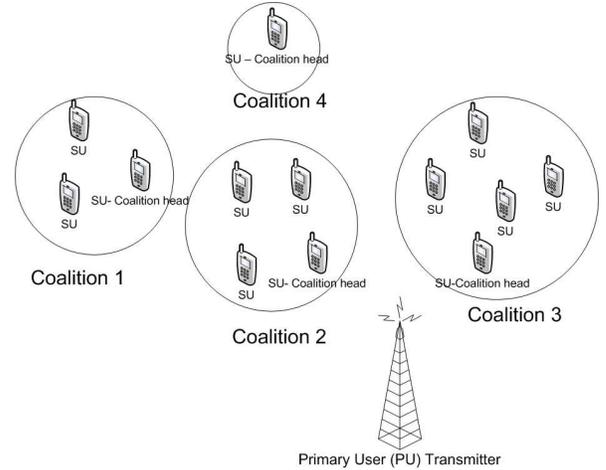}
\end{center}\vspace{-0.5cm}
\caption {An illustrative example of coalition formation for collaborative spectrum sensing among SUs.} \label{fig:ill}\vspace{-0.6cm}
\end{figure}
In order to minimize their missing probabilities, the SUs will interact for forming coalitions of collaborating SUs. Within each coalition $S \subseteq \mathcal{N}=\{1,\ldots,N\}$, a SU, selected as \emph{coalition head}, collects the sensing bits from the coalition's SUs and acts as a fusion center in order to make a coalition-based decision on the presence or absence of the PU. This can be seen as having the centralized collaborative sensing of \cite{CS00}, \cite{CS03} applied at the level of each coalition with the coalition head being the fusion center to which all the coalition members report. For combining the sensing bits and making the final detection decision, the coalition head will use the decision fusion OR-rule. Within each coalition we take into account the probability of error due to the fading on the reporting channel between the SUs of a coalition and the coalition head \cite{CS03}. Inside a coalition $S$, assuming BPSK modulation in Rayleigh fading environments, the probability of reporting error between a SU $i \in S$ and the coalition head $k \in S$ is given by \cite{PROAKIS}
 \begin{equation}\label{eq:err}
P_{e,i,k} = \frac{1}{2} \left(1-\sqrt{\frac{\bar{\gamma}_{i,k}}{1+ \bar{\gamma}_{i,k}}}\right),
\end{equation}
where $\bar{\gamma}_{i,k}=\frac{P_i h_{i,k}}{\sigma^2}$ is the average SNR for bit reporting between SU $i$ and the coalition head $k$ inside $S$ with $P_i$ the transmit power of SU $i$ used for reporting the sensing bit to $k$ and $h_{i,k}=\frac{\kappa}{d_{i,k}^{\mu}}$ the path loss between SU $i$ and coalition head $k$. Any SU can be chosen as a coalition head within a coalition. However, for the remainder of this paper, we adopt the following convention without loss of generality.
\begin{convention}
Within a coalition $S$, the SU $k \in S$ having the lowest non-cooperative missing probability $P_{m,k}$ is chosen as \emph{coalition head}. Hence, the coalition head $k$ of a coalition $S$ is given by $k= \underset{i\in S}{\operatorname{arg\,min}}\,P_{m,i}$ with $P_{m,i}$ given by (\ref{eq:missprob}).
\end{convention}
The motivation behind Convention~1 is that the SU having the lowest missing probability (best detection probability) within a coalition should not risk sending his local sensing bit over the fading reporting channel; and thus it will serve as a fusion center for the other SUs in the coalition. By collaborative sensing, the missing and false alarm probabilities of a coalition $S$ having coalition head $k$ are, respectively, given by \cite{CS03}
\begin{equation}\label{eq:missclu}
Q_{m,S} = \displaystyle \prod_{i \in S} [ P_{m,i}(1-P_{e,i,k}) + (1- P_{m,i})P_{e,i,k}],
\end{equation}
\begin{equation}\label{eq:falseclu}
Q_{f,S} =  1 - \displaystyle\prod_{i \in S} [ (1-P_{f})(1-P_{e,i,k}) + P_{f}P_{e,i,k}],
\end{equation}
where $P_{f}$, $P_{m,i}$ and $P_{e,i,k}$ are respectively given by (\ref{eq:falseind}), (\ref{eq:missprob}) and (\ref{eq:err}) for a SU $i \in S$ and coalition head $k \in S$.

It is clear from (\ref{eq:missclu}) and (\ref{eq:falseclu}) that as the number of SUs per coalition increases, the missing probability will decrease while the probability of false alarm will increase. This is a crucial trade off in collaborative spectrum sensing that can have a major impact on the collaboration strategies of each SU. Thus, our objective is to derive distributed strategies allowing the SUs to collaborate while accounting for this trade off. An example of the sought network structure is shown in Fig.~\ref{fig:ill}.

\section{Collaborative Spectrum Sensing As Coalitional Game}\label{sec:coal}
In this section, we model the problem of collaborative spectrum sensing as a coalitional game. Then we
prove and discuss its key properties.
\subsection{Centralized Approach}
A centralized approach can be used in order to find the optimal coalition structure, such as in Fig.\ref{fig:ill}, that allows the SUs to maximize their benefits from collaborative spectrum sensing. For instance, we seek a centralized solution that minimizes the average missing probability (maximizes the average detection probability) per SU subject to a false alarm probability constraint per SU. In  a centralized approach, we assume the existence of a centralized entity in the network that is able to gather information on the SUs such as their individual missing probabilities or their location. In brief, the centralized entity must be able to know all the required parameters for computing the probabilities in (\ref{eq:missclu}) and (\ref{eq:falseclu}) in order to find the optimal structure. However, prior to deriving such an optimal centralized solution, the following property must be pinpointed within each coalition.
\begin{property}
The missing and false alarm probabilities of any SU $i \in S$ are given by the missing and false alarm probabilities of the coalition $S$ in (\ref{eq:missclu}) and (\ref{eq:falseclu}), respectively.
\end{property}
\begin{proof}
Within each coalition $S$ the SUs report their sensing bits to the coalition head. In its turn the coalition head of $S$ combines the sensing bits using decision fusion and makes a final decision on the presence or absence of the PU. Thus, SUs belonging to a coalition $S$ will transmit or not based on the final coalition head decision. Consequently, the missing and false alarm probabilities of any SU $i \in S$ are the missing and false alarm probabilities of the coalition $S$ to which $i$ belongs as given by in (\ref{eq:missclu}) and (\ref{eq:falseclu}), respectively.
\end{proof}

As a consequence of Property~1 the required false alarm probability constraint per SU directly maps to a false alarm probability constraint \emph{per coalition}. Therefore, denoting $\mathcal{B}$ as the set of \emph{all partitions} of $\mathcal{N}$, the centralized approach seeks to solve the following optimization problem
 \begin{equation}\label{eq:opt}
\min_{\mathcal{P} \in \mathcal{B}}{\frac{\sum_{S \in \mathcal{P}}|S|\cdot Q_{m,S}}{N} },
 \end{equation}
\begin{equation*}
\textrm{ s.t. }Q_{f,S}\le \alpha \ \forall\ S \in \mathcal{P},
 \end{equation*}
 where $|\cdot|$ represents the cardinality of a set operator and $S$ is a coalition belonging to the partition $\mathcal{P}$. Clearly, the centralized optimization problem seeks to find the optimal partition $\mathcal{P}^{*} \in \mathcal{B}$ that minimizes the average missing probability per SU, subject to a false alarm constraint per SU (coalition).

However, it is shown in \cite{NP00} that finding the optimal coalition structure for solving an optimization problem such as in  (\ref{eq:opt}) is an NP-complete problem. This is mainly due to the fact that the number of possible coalition structures (partitions), given by the Bell number, grows exponentially with the number of SUs $N$ \cite{NP00}. Moreover, the complexity increases further due to the fact that the expressions of $Q_{m,S}$ and $Q_{f,S}$ given by (\ref{eq:missclu}) and (\ref{eq:falseclu}) depend on the optimization parameter $\mathcal{P}$. For this purpose, deriving a distributed solution enabling the SUs to benefit from collaborative spectrum sensing with a low complexity is desirable. The above formulated centralized approach will be used as a benchmark for the distributed solution in the simulations, for reasonably small networks.

\subsection{Game Formulation and Properties}
For the purpose of deriving a distributed algorithm that can minimize the missing probability per SU, we refer to cooperative game theory \cite{Game_theory2} which provides a set of analytical tools suitable for such algorithms. For instance, the proposed collaborative sensing problem can be modeled as a $(\mathcal{N},v)$ coalitional game \cite{Game_theory2} where $\mathcal{N}$ is the set of players (the SUs) and $v$ is the utility function or value of a coalition.

The value $v(S)$ of a coalition $S \subseteq \mathcal{N}$ must capture the trade off between the probability of detection and the probability of false alarm. For this purpose, $v(S)$ must be an increasing function of the detection probability $Q_{d,S} = 1 - Q_{m,S}$ within coalition $S$ and a decreasing function of the false alarm probability $Q_{f,S}$. A suitable utility function is given by
\begin{equation}\label{eq:util}
v(S) =  Q_{d,S} - C(Q_{f,S}) =  (1 - Q_{m,S}) - C(Q_{f,S}),
\end{equation}
where $Q_{m,S}$ is the missing probability of coalition $S$ given by (\ref{eq:missclu}) and $C(Q_{f,S})$ is a cost function of the false alarm probability within coalition $S$ given by (\ref{eq:falseclu}).

First of all, we provide the following definition from \cite{Game_theory2} and subsequently prove an interesting property pertaining to the proposed game model.
\begin{definition}
A coalitional game $(\mathcal{N},v)$ is said to have a \emph{transferable} utility if the value $v(S)$ can be arbitrarily apportioned between the coalition's players. Otherwise, the coalitional game has a \emph{non-transferable utility} and each player will have its own utility within coalition $S$.
\end{definition}
\begin{property}
In the proposed collaborative sensing game, the utility of a coalition $S$ is equal to the utility of each SU in the coalition, i.e., $v(S) = \phi_i(S), \ \forall\ i \in S$, where $\phi_i(S)$ denotes the utility of SU $i$ when $i$ belongs to a coalition $S$. Consequently, the proposed $(\mathcal{N},v)$ coalitional game model has a \emph{non-transferable} utility.
\end{property}
\begin{proof}
The coalition value in the proposed game is given by (\ref{eq:util}) and is a function of $Q_{m,S}$ and $Q_{f,S}$. As per Property~1, the missing probabilities for each SU $i \in S$ are also given by $Q_{m,S}$ and $Q_{f,S}$ and, thus, the utility of each SU $i \in S$ is $\phi_i(S) = v(S)$. Hence, the coalition value $v(S)$ cannot be arbitrarily apportioned among the users of a coalition; and the proposed coalitional game has non-transferable utility.
\end{proof}

In general, coalitional game based problems seek to characterize the properties and stability of the grand coalition of all players since it is generally assumed that the grand coalition maximizes the utilities of the players \cite{Game_theory2}. In our case, although collaborative spectrum sensing improves the detection probability for the SUs; the cost in terms of false alarm limits this gain. Therefore, for the proposed $(\mathcal{N},v)$ coalitional game we have the following property.
\begin{property}
For the proposed $(\mathcal{N},v)$ coalitional game, the grand coalition of all the SUs does not always form due to the collaboration false alarm costs; thus disjoint independent coalitions will form in the network.
\end{property}
\begin{proof}
By inspecting $Q_{m,S}$ in (\ref{eq:missclu}) and through the results shown in \cite{CS03} it is clear that as the number of SUs in a coalition increase $Q_{m,S}$ decreases and the performance in terms of detection probability improves. Hence, when no cost for collaboration exists, the grand coalition of all SUs is the optimal structure for maximizing the detection probability. However, when the number of SUs in a coalition $S$ increases, it is shown in \cite{CS03} through (\ref{eq:missclu}) that the false alarm probability increases. Therefore, for the proposed collaborative spectrum sensing model with cost for collaboration, the grand coalition of all SUs will, in general, not form due to the false alarm cost as taken into consideration in (\ref{eq:util}).
\end{proof}

In a nutshell, we have a non-transferable $(\mathcal{N},v)$ coalitional game and we seek to derive a distributed algorithm for forming coalitions among SUs. Before deriving such an algorithm, we will delve into the details of the cost function in (\ref{eq:util}).\vspace{-0.14cm}
\subsection{Cost Function}
Any well designed cost function $C(Q_{f,S})$ in (\ref{eq:util}) must satisfy several requirements needed for adequately modeling the false alarm cost. On one hand, $C(Q_{f,S})$ must be an increasing function of $Q_{f,S}$ with the increase slope becoming steeper as $Q_{f,S}$ increases. On the other hand, the cost function $C(Q_{f,S})$ must impose a maximum tolerable false alarm probability, i.e. an upper bound constraint on the false alarm, that cannot be exceeded by any SU in a manner similar to the centralized problem in (\ref{eq:opt}) (due to Property~1, imposing a false alarm constraint on the coalition maps to a constraint per SU).

A well suited cost function satisfying the above requirements is the logarithmic barrier penalty function given by \cite{BO00}
\begin{equation}\label{eq:logbarr}
C(Q_{f,S}) = \begin{cases}- \alpha^2\cdot\log{\left(1-\left(\frac{Q_{f,S}}{\alpha}\right)^2\right)}, & \mbox{if } Q_{f,S} < \alpha,\\ +\infty, &
\mbox{if } Q_{f,S} \ge \alpha, \end{cases}
\end{equation}
where $\log$ is the natural logarithm and $\alpha$ is a false alarm constraint per coalition (per SU). The cost function in (\ref{eq:logbarr}) allows to incur a penalty which is increasing with the false alarm probability. Moreover, it imposes a maximum false alarm probability per SU. In addition, as the false alarm probability gets closer to $\alpha$ the cost for collaboration increases steeply, requiring a significant improvement in detection probability if the SUs wish to collaborate as per (\ref{eq:util}). Also, it is interesting to note that the proposed cost function depends on both distance and the number of SUs in the coalition, through the false alarm probability $Q_{f,S}$ (the distance lies within the probability of error). Hence, the cost for collaboration increases with the number of SUs in the coalition as well as when the distance between the coalition's SUs increases.

\section{Distributed Coalition Formation Algorithm}\label{sec:coalform}
In this section, we propose a distributed coalition formation algorithm and we discuss its main properties.
\subsection{Coalition Formation Concepts}
Coalition formation has been a topic of high interest in game
theory \cite{NP00,DM00,KA01,KA02}. The goal is to find
algorithms for characterizing the coalitional structures that form in a network where the grand coalition is not optimal. For instance, a generic framework for coalition formation is presented in \cite{KA00,KA01,KA02} whereby coalitions form and break through
two simple merge-and-split rules. This framework can be used to construct
a distributed coalition formation algorithm for collaborative sensing, but first, we define the following concepts \cite{KA01,KA02}.

\begin{definition}
A \emph{collection} of coalitions, denoted  $\mathcal{S}$, is defined as the set $\mathcal{S} = \{S_{1},\ldots,S_{l}\}$ of mutually disjoint coalitions $S_{i} \subset \mathcal{N}$. If the collection spans \emph{all} the players of $ \mathcal{N}$; that is $\bigcup_{j=1}^{l} S_j =  \mathcal{N}$, the collection is a \emph{partition} of $\mathcal{N}$. 
\end{definition}
\begin{definition}
A preference operator or \emph{comparison relation} $\rhd$ is defined for comparing two collections $\mathcal{R} = \{R_{1},\ldots,R_{l}\}$ and
$\mathcal{S} = \{S_{1},\ldots,S_{m}\}$ that are partitions of the same subset $\mathcal{A} \subseteq \mathcal{N}$ (same players in $\mathcal{R}$ and $\mathcal{S}$). Thus, $\mathcal{R} \rhd \mathcal{S}$ implies that the way $\mathcal{R}$ partitions $\mathcal{A}$ is preferred to the way $\mathcal{S}$ partitions $\mathcal{A}$ based on a criterion to be defined next.
\end{definition}

Various criteria (referred to as \emph{orders}) can be used as comparison relations between collections or partitions \cite{KA01},\cite{KA02}. These orders are divided into two main categories: coalition value orders and individual value orders. Coalition value orders compare two collections (or partitions) using the value of the coalitions inside these collections such as in the utilitarian order where $\mathcal{R} \rhd \mathcal{S}$ implies $\sum_{i=1}^{l}v(R_i) > \sum_{i=1}^{m}v(S_i)$. Individual value orders perform the comparison using the actual player utilities and not the coalition value. For such orders, two collections $\mathcal{R}$ and $\mathcal{S}$ are seen as sets of player utilities of the same length $L$ (number of players). The players' utilities are either the payoffs after division of the value of the coalitions in a collection (transferable utility) or the actual utilities of the players belonging to the coalitions in a collection (non-transferable utility). Due to the non-transferable nature of the proposed $(\mathcal{N},v)$ collaborative sensing game (Property~2), an individual value order must be used as a comparison relation $\rhd$. An important example of individual value orders is the \emph{Pareto order}. Denote for a collection $\mathcal{R} = \{R_{1},\ldots,R_{l}\}$, the utility of a player $j$ in a coalition $R_j \in \mathcal{R}$ by $\phi_j(\mathcal{R})=\phi_j(R_j)=v(R_j)$ (as per Property~2); hence, the Pareto order is defined as follows
\begin{align}\label{eq:Pareto}
\mathcal{R} \rhd \mathcal{S} \Longleftrightarrow \{\phi_{j}(\mathcal{R}) \ge
\phi_{j}(\mathcal{S}) \  \forall \ j \in \mathcal{R},\mathcal{S}\},   \\ \textrm{ with \emph{at
least one strict inequality} ($>$) for a player } k. \nonumber
\end{align}
Due to the non-transferable nature of the proposed collaborative sensing model, the Pareto order is an adequate preference relation. Having defined the various concepts, we derive a distributed coalition formation algorithm in the next subsection.
\subsection{Coalition Formation Algorithm}\label{sec:mergeandsplit}
For autonomous coalition formation in cognitive radio networks, we propose a distributed algorithm based on two simple rules denoted as ``merge'' and ``split'' that
allow to modify a partition $\mathcal{T}$ of the SUs set $\mathcal{N}$ as follows \cite{KA01}.
\begin{definition}
\textbf{Merge Rule -} Merge any set of coalitions
$\{S_{1},\ldots,S_{l}\}$ where $\{\bigcup_{j=1}^{l}S_{j}\} \rhd \{S_{1},\ldots,S_{l}\} $, therefore, {$\{S_{1},\ldots,S_{l}\}
\rightarrow \{\bigcup_{j=1}^{l}S_{j}\}$}, (each $S_i$ is a coalition in $\mathcal{T}$).
\end{definition}
\begin{definition}
\textbf{Split Rule -} Split any coalition
$\bigcup_{j=1}^{l}S_{j}$ where  $\{S_{1},\ldots,S_{l}\} \rhd \{\bigcup_{j=1}^{l}S_{j}\} $, thus, $\{\bigcup_{j=1}^{l}S_{j}\}
\rightarrow \{S_{1},\ldots,S_{l}\}$, (each $S_i$ is a coalition in $\mathcal{T}$).
\end{definition}
Using the above rules, multiple coalitions can merge into a larger coalition if merging yields a preferred collection based on the selected order $\rhd$. Similarly, a coalition would split into smaller coalitions if splitting yields a preferred collection. When $\rhd$ is the Pareto order, coalitions will merge (split) only if at least one SU is able to strictly improve its individual utility through this merge (split) without decreasing the other SUs' utilities.
%
By using the merge-and-split rules combined with the Pareto order, a distributed coalition formation algorithm suited for collaborative spectrum sensing can be constructed. First and foremost, the appeal of forming coalitions using merge-and-split stems from the fact that it has been shown in \cite{KA01} and \cite{KA02} that any arbitrary iteration of
merge-and-split operations \emph{terminates}.  Moreover, each merge or split decision can be taken in a distributed manner by each individual SU or by each already formed coalition. Subsequently, a merge-and-split coalition algorithm can adequately model the distributed interactions among the SUs of a cognitive network that are seeking to collaborate in the sensing process.

In consequence, for the proposed collaborative sensing game, we construct a coalition formation algorithm based on merge-and-split and divided into three phases: local sensing, adaptive coalition formation, and coalition sensing. In the local sensing phase, each individual SU computes its own local PU detection bit based on the received PU signal. In the adaptive coalition formation phase, the SUs (or existing coalitions of SUs) interact in order to assess whether to share their sensing results with nearby coalitions. For this purpose, an iteration of sequential merge-and-split rules occurs in the network, whereby each coalition decides to merge (or split) depending on the utility improvement that merging (or splitting) yields. In the final coalition sensing phase, once the network topology converges following merge-and-split, SUs that belong to the same coalition report their local sensing bits to their local coalition head. The coalition head subsequently uses decision fusion OR-rule to make a final decision on the presence or the absence of the PU. This decision is then reported by the coalition heads to all the SUs within their respective coalitions.
\indent Each round of the three phases of the proposed algorithm starts from an initial network partition $\mathcal{T} =\{T_1,\ldots,T_l\}$ of $\mathcal{N}$. During the adaptive coalition formation phase any random coalition (individual SU) can start with the merge process. For implementation purposes, assume that the coalition  $T_i \in  \mathcal{T}$ which has the highest utility in the initial partition $\mathcal{T}$ starts the merge by attempting to collaborate with a nearby coalition. On one hand, if merging occurs, a new coalition $\tilde{T}_i$ is formed and, in its turn, coalition $\tilde{T}_i$ will attempt to merge with a nearby SU that can improve its utility. On the other hand, if $T_i$ is unable to merge with the firstly discovered partner, it tries to find other coalitions that have a mutual benefit in merging. The search ends by a final merged coalition $T_i^{\text{final}}$ composed of $T_i$ and one or several of coalitions in its vicinity ($T_i^{\text{final}} = T_i$, if no merge occurred). The algorithm is repeated for the remaining $T_i \in \mathcal{T}$ until all the coalitions have made their merge decisions, resulting in a final partition $\mathcal{F}$. Following the merge process, the coalitions in the resulting partition $\mathcal{F}$ are next subject to split operations, if any is possible. An iteration consisting of multiple successive merge-and-split operations is repeated until it terminates. It must be stressed that the decisions to merge or split can be taken in a distributed way without relying on any centralized entity as each SU or coalition can make its own decision for merging or splitting. Table~\ref{tab:alg1} summarizes one round of the proposed algorithm.

For handling environmental changes such as mobility or the joining/leaving of SUs, Phase~2 of the proposed algorithm in Table~\ref{tab:alg1} is repeated periodically. In Phase~2, periodically, as time evolves and SUs (or the PU) move or join/leave,  the SUs can autonomously self-organize and adapt the network's topology through new merge-and-split iterations with each coalition taking the decision to merge (or split) subject to satisfying the merge (or split) rule through Pareto order (\ref{eq:Pareto}). In other words, every period of time $\theta$ the SUs assess the possibility of splitting into smaller coalitions or merging with new partners. The period $\theta$ is smaller in highly mobile environments to allow a more adequate adaptation of the topology. Similarly, every period $\theta$, in the event where the current coalition head of a coalition has moved or is turned off, the coalition members may select a new coalition head if needed. The convergence of this merge-and-split adaptation to environmental changes is always guaranteed, since, by definition of the merge and split rules, any iteration of these rules certainly terminates.

For the proposed coalition formation algorithm, an upper bound on the maximum coalition size is imposed by the proposed utility and cost models in (\ref{eq:util}) and (\ref{eq:logbarr}) as follows:
\begin{theorem}
For the proposed collaborative sensing model, any coalition structure resulting from the distributed coalition formation algorithm will have coalitions limited in size to a maximum of $M_{\max}=\frac{\log{(1-\alpha)}}{\log{(1-P_f)}}$ SUs.
\end{theorem}
\begin{proof}
For forming coalitions, the proposed algorithm requires an improvement in the utility of the SUs through Pareto order. However, the benefit from collaboration is limited by the false alarm probability cost modeled by the barrier function (\ref{eq:logbarr}). A minimum false alarm cost in a coalition  $S$ with coalition head $k\in S$ exists whenever the reporting channel is perfect, i.e., exhibiting no error, hence $P_{e,i,k} = 0\ \forall i\in S$. In this perfect case, the false alarm probability in a perfect coalition $S_p$ is given by
\begin{equation}\label{eq:perf}
Q_{f,S_p} =  1 - \displaystyle\prod_{i \in S_p}  (1-P_{f}) = 1- (1-P_f)^{|S_p|},
\end{equation}
where $|S_p|$ is the number of SUs in the perfect coalition $S_p$. A perfect coalition $S_p$ where the reporting channels inside are perfect (i.e. SUs are grouped very close to each other) can accommodate the largest number of SUs relative to other coalitions. Hence, we can use this perfect coalition to find an upper bound on the maximum number of SUs per coalition. For instance, the log barrier function in (\ref{eq:logbarr}) tends to infinity whenever the false alarm probability constraint per coalition is reached which implies an upper bound on the maximum number of SUs per coalition if
$Q_{f,S_p}\ge \alpha,$
yielding by (\ref{eq:perf})
\begin{equation}
|S_p| \le \frac{\log{(1-\alpha)}}{\log{(1-P_f)}} = M_{\max}\ .
\end{equation}
\end{proof}
 \begin{table}[!t]
  \centering
  \caption{
    \vspace*{-0.7em}One round of the proposed collaborative sensing algorithm}\vspace*{-1em}
    \begin{tabular}{p{8cm}}
      \hline
      \textbf{Initial State} \vspace*{.5em} \\
      \hspace*{1em} The network is partitioned by $\mathcal{T}=\{T_1,\ldots,T_k\}$ (At the beginning of all time $\mathcal{T}$ = $\mathcal{N}$ = $\{1,\ldots,N\}$ with non-cooperative SUs). \vspace*{.1em}\\
\textbf{Three phases in each round of the coalition formation algorithm} \vspace*{.5em}\\
\hspace*{1em}\emph{Phase 1 - Local Sensing:}   \vspace*{.1em}\\
\hspace*{2em}Each individual SU computes its local PU signal sensing bit.\vspace*{.1em}\\
\hspace*{1em}\emph{Phase 2 - Adaptive Coalition Formation:}   \vspace*{.1em}\\
\hspace*{2em}In this phase, coalition formation using merge-and-split occurs. \vspace*{.2em}\\
\hspace*{3em}\textbf{repeat}\vspace*{.2em}\\
\hspace*{4em}a) $\mathcal{F}$ = Merge($\mathcal{T}$); coalitions in $\mathcal{T}$ decide to merge\\ \hspace*{4em}based on the merge algorithm explained in Section~\ref{sec:mergeandsplit}.
\hspace*{4em}b) $\mathcal{T}$ = Split($\mathcal{F}$); coalitions in $\mathcal{F}$ decide to split based on \hspace*{4em}the Pareto order.\\
\hspace*{3em}\textbf{until} merge-and-split terminates.\vspace*{.2em}\\
\hspace*{1em}\emph{Phase 3 - Coalition Sensing:}   \vspace*{.1em}\\
\hspace*{2em}a) Each SU reports its sensing bit to the coalition head.\vspace*{.1em}\\
\hspace*{2em}b) The coalition head of each coalition makes a final decision on\vspace*{.1em}\\
\hspace*{2em}the absence or presence of he PU using decision fusion OR-rule.\vspace*{.1em}\\
\hspace*{2em}c) The SUs in a coalition abide by the final decision made by the\vspace*{.1em}\\
\hspace*{2em}coalition head.\vspace*{.1em}\\
\textbf{The above phases are repeated throughout the network operation. In Phase~2, through distributed and periodic merge-and-split decisions, the SUs can autonomously adapt the network topology to environmental changes such as mobility. } \vspace*{.5em}\\
   \hline
    \end{tabular}\label{tab:alg1}\vspace{-0.5cm}
\end{table}

It is interesting to note that the maximum size of a coalition $M_{\max}$ depends mainly on two parameters: the false alarm constraint $\alpha$ and the non-cooperative false alarm $P_f$. For instance, larger false alarm constraints allow larger coalitions, as the maximum tolerable cost limit for collaboration is increased. Moreover, as the non-cooperative false alarm $P_f$ decreases, the possibilities for collaboration are better since the increase of the false alarm due to coalition size becomes smaller as per (\ref{eq:falseclu}).  It must be noted that the dependence of $M_{\max}$ on $P_f$ yields a direct dependence of $M_{\max}$ on the energy detection threshold $\lambda$ as per (\ref{eq:falseind}). Finally, it is interesting to see that the upper bound on the coalition size does not depend on the location of the SUs in the network nor on the actual number of SUs in the network. Therefore, deploying more SUs or moving the SUs in the network for a fixed $\alpha$ and $P_f$ does not increase the upper bound on coalition size.

\subsection{Stability}

The result of the proposed algorithm in Table~\ref{tab:alg1} is a
network partition composed of disjoint independent coalitions of SUs. The stability of this resulting network structure can be investigated using the concept of a defection function $\mathbb{D}$ \cite{KA01}.
\begin{definition}
A \emph{defection} function $\mathbb{D}$ is a function which
associates with each partition $\mathcal{T}$ of $\mathcal{N}$ a group of
collections in $\mathcal{N}$. A partition $\mathcal{T} = \{T_1,\ldots,T_l\}$ of $\mathcal{N}$ is \emph{$\mathbb{D}$-stable} if no group of players is interested in leaving $\mathcal{T}$ when the players who leave can only form the collections allowed by $\mathbb{D}$.
\end{definition}

Two important defection functions must be characterized  \cite{KA00,KA01,KA02}. First, the $\mathbb{D}_{hp}(\mathcal{T})$
function (denoted $\mathbb{D}_{hp}$) which associates with each partition
$\mathcal{T}$ of $\mathcal{N}$ the group of all partitions of $N$ that the players
can form through merge-and-split operations applied to $\mathcal{T}$. This
function allows any group of players to leave the partition $\mathcal{T}$ of
$\mathcal{N}$ through merge-and-split operations to create another
\emph{partition} in $\mathcal{N}$. Second, the $\mathbb{D}_{c}(\mathcal{T})$ function
(denoted $\mathbb{D}_{c}$) which associates with each partition $\mathcal{T}$ of $\mathcal{N}$
the family of all collections in $\mathcal{N}$. This function allows any
group of players to leave the partition $\mathcal{T}$ of $\mathcal{N}$ through
\emph{any} operation and create an arbitrary \emph{collection} in
$\mathcal{N}$. Two forms of stability stem from these definitions:
$\mathbb{D}_{hp}$ stability and a stronger $\mathbb{D}_{c}$
stability. A partition $\mathcal{T}$ is $\mathbb{D}_{hp}$-stable, if no
players in $\mathcal{T}$ are interested in leaving $\mathcal{T}$ through
merge-and-split to form other partitions in $\mathcal{N}$; while a partition
$\mathcal{T}$ is $\mathbb{D}_{c}$-stable, if no players in $\mathcal{T}$ are
interested in leaving $\mathcal{T}$ through \emph{any} operation (not
necessary merge or split) to form other collections in $\mathcal{N}$.

Characterizing any type of $\mathbb{D}$-stability for a partition depends on various properties of its coalitions. For instance, a partition $\mathcal{T}=\{T_1,\ldots,T_l\}$ is $\mathbb{D}_{hp}$-stable if, for the partition $\mathcal{T}$, no coalition has an incentive to split or merge. As an immediate result of the definition of $\mathbb{D}_{hp}$-stability we have

\begin{theorem}
Every partition resulting from our proposed coalition formation algorithm is $\mathbb{D}_{hp}$-stable.
\end{theorem}
Briefly, a $\mathbb{D}_{hp}$-stable can be thought of as a state of equilibrium where no coalitions have an incentive to pursue coalition formation through merge or split.
With regards to $\mathbb{D}_{c}$ stability, the work in \cite{KA00,KA01,KA02} proved that a $\mathbb{D}_{c}$-stable partition has the following properties:
\begin{enumerate}
\item If it exists, a $\mathbb{D}_{c}$-stable partition is the
\emph{unique} outcome of any \emph{arbitrary} iteration of
merge-and-split and is a $\mathbb{D}_{hp}$-stable
partition. \item A $\mathbb{D}_{c}$-stable partition $\mathcal{T}$ is a
unique $\rhd$-maximal partition, that is for all partitions $ \mathcal{T}'
\neq \mathcal{T}$ of $\mathcal{N}$, $\mathcal{T} \rhd \mathcal{T}'$. In the case where $\rhd$ represents
the Pareto order, this implies that the $\mathbb{D}_{c}$-stable
partition $\mathcal{T}$ is the partition that presents a \emph{Pareto
optimal} utility distribution for all the players.
\end{enumerate}

However, the existence of a $\mathbb{D}_{c}$-stable partition is not always guaranteed \cite{KA01}. The  $\mathbb{D}_{c}$-stable partition $\mathcal{T} = \{T_{1},\ldots,T_{l} \}$ of the whole space $\mathcal{N}$ exists if a partition of $\mathcal{N}$ that verifies the following two necessary and sufficient conditions exists\cite{KA01}:
\begin{enumerate}
\item For each $i\in \{1,\ldots,l\}$ and each pair of disjoint \emph{coalitions}  $S_1$ and $S_2$ such that $\{S_1 \cup S_2\} \subseteq T_i$ we have
$\{S_1 \cup S_2\} \rhd \{S_1,S_2\}$.
\item For the partition $\mathcal{T}=\{T_1,\ldots,T_l\}$ a coalition $G \subset \mathcal{N}$ formed of players belonging to different $T_i \in \mathcal{T}$ is $\mathcal{T}$-incompatible if for no
$i \in \{1,\ldots,l\}$ we have $G\subset T_i$. $\mathbb{D}_{c}$-stability requires that for  all $\mathcal{T}$-incompatible
coalitions $\{G\}[\mathcal{T}] \rhd \{G\}$ where $\{G\}[\mathcal{T}] = \{G\cap T_i \ \forall \ i\in\{1,\ldots,l\}\}$ is the projection of coalition $G$ on $\mathcal{T}$.

\end{enumerate}
If no partition of $\mathcal{N}$ can satisfy these conditions, then no $\mathbb{D}_{c}$-stable partitions of $\mathcal{N}$ exists. Nevertheless, we have
\begin{lemma}
For the proposed $(\mathcal{N},v)$ collaborative sensing coalitional game, the proposed algorithm of Table~\ref{tab:alg1} converges to the optimal $\mathbb{D}_{c}$-stable partition, if such a partition exists. Otherwise, the proposed algorithm yields a final network partition that is $\mathbb{D}_{hp}$-stable.
\end{lemma}
\begin{proof}
The proof is an immediate consequence of Theorem~2 and the fact that the  $\mathbb{D}_{c}$-stable partition is a unique outcome of any arbitrary merge-and-split iteration which is the case with any partition resulting from our algorithm.
\end{proof}
\begin{figure}[t]
\begin{center}
\includegraphics[angle=0,width=\columnwidth]{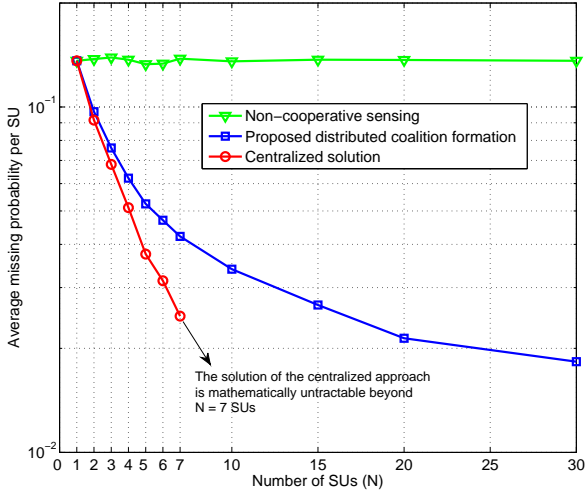}
\end{center}\vspace{-0.7cm}
\caption {Average missing probabilities (average over locations of SUs and non-cooperative false alarm range $P_f \in (0,\alpha)$ ) vs. number of SUs.} \label{fig:perf}\vspace{-0.7cm}
\end{figure}

Moreover, for the proposed game, the existence of the $\mathbb{D}_{c}$-stable partition cannot be always guaranteed. For instance, for verifying the first condition for existence of the  $\mathbb{D}_{c}$-stable partition, the SUs belonging to partitions of each coalitions must verify the Pareto order through their utility given by (\ref{eq:util}). Similarly, for verifying the second condition of $\mathbb{D}_{c}$ stability,  SUs belonging to all $\mathcal{T}$-incompatible coalitions in the network must verify the Pareto order. Consequently, finding a geometrical closed-form condition for the existence of such a partition is not feasible as it depends on the location of the SUs and the PU through the individual missing and false alarm probabilities in the utility expression (\ref{eq:util}). Hence, the existence of the $\mathbb{D}_{c}$-stable partition  is closely tied to the location of the SUs and the PU which both can be random parameters in practical networks. However, the proposed algorithm will always guarantee convergence to this optimal  $\mathbb{D}_{c}$-stable partition when it exists as stated in Lemma~1. Whenever a $\mathbb{D}_{c}$-stable partition does not exist, the coalition structure resulting from the proposed algorithm will be $\mathbb{D}_{hp}$-stable (no coalition or SU is able to merge or split any further).\vspace{-2.5mm}

\section{Simulation Results and Analysis}\label{sec:sim}
\begin{figure}[t]
\begin{center}
\includegraphics[angle=0,width=\columnwidth]{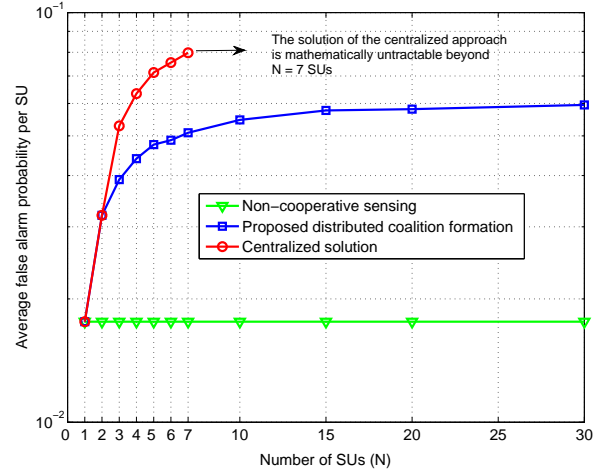}
\end{center}\vspace{-0.5cm}
\caption {Average false alarm probabilities (average over locations of SUs and non-cooperative false alarm range $P_f \in (0,\alpha)$ ) vs. number of SUs.} \label{fig:falarm}\vspace{-0.6cm}
\end{figure}
For simulations, the following network is set up: The PU is
placed at the origin of a $3$~km $\times 3$~km square with the SUs randomly deployed in the area around the PU. We set the time bandwidth product $m=5$ \cite{CS00,CS01,CS03}, the PU transmit power $P_{\textrm{PU}}=100$~mW, the SU transmit power for reporting the sensing bit $P_i=10$~mW $\forall i \in \mathcal{N}$ and the noise level $\sigma^2=-90$~dBm. For path loss, we set $\mu=3$ and $\kappa = 1$. The maximum false alarm constraint is set to $\alpha = 0.1$, as recommended by the IEEE 802.22 standard \cite{RE00}.

In Figs.~\ref{fig:perf} and \ref{fig:falarm} we show, respectively, the average missing probabilities and the average false alarm probabilities achieved per SU for different network sizes. These probabilities are averaged over random locations of the SUs as well as a range of energy detection thresholds $\lambda$ that do not violate the false alarm constraint; this in turn, maps into an average over the non-cooperative false alarm range $P_f \in (0,\alpha)$ (obviously, for $P_f > \alpha$ no cooperation is possible). In Fig.~\ref{fig:perf}, we show that the proposed algorithm yields a significant improvement in the average missing probability reaching up to $86.6\%$ reduction (at $N=30$) compared to the non-cooperative case. This advantage is increasing with the network size $N$. However, there exists a gap in the performance of the proposed algorithm and that of the optimal centralized solution. This gap stems mainly from the fact that the log barrier function used in the distributed algorithm (\ref{eq:logbarr}) increases the cost drastically when the false alarm probability is in the vicinity of $\alpha$. This increased cost makes it harder for coalitions with false alarm levels close to $\alpha$ to collaborate in the distributed approach as they require a large missing probability improvement to compensate the cost in their utility (\ref{eq:util}) so that a Pareto order merge or split becomes possible. However, albeit the proposed cost function yields a performance gap in terms of missing probability,  it forces a false alarm for the distributed case smaller than that of the centralized solution as seen in Fig.~\ref{fig:falarm}.

For instance, Fig.~\ref{fig:falarm} shows that the achieved average false alarm by the proposed distributed solution outperforms that of the centralized solution but is still outperformed by the non-cooperative case. Thus, while the centralized solution achieves a better missing probability; the proposed distributed algorithm compensates this performance gap through the average achieved false alarm. In summary, Figs.~\ref{fig:perf} and \ref{fig:falarm} clearly show the performance trade off that exists between the gains achieved by collaborative spectrum sensing in terms of average missing probability and the corresponding cost in terms of average false alarm probability.

In Fig.~\ref{fig:perffal}, we show the average missing probabilities per SU for different energy detection thresholds $\lambda$ expressed by the feasible range of non-cooperative false alarm probabilities $P_f \in (0,\alpha)$ for $N=7$. In this figure, we show that as the non-cooperative $P_f$ decreases the performance advantage of collaborative spectrum sensing for both the centralized and distributed solutions increases (except for very small $P_f$ where the advantage in terms of missing probability reaches its maximum). The performance gap between centralized and distributed is once again compensated by a false alarm advantage for the distributed solution as already seen and explained in Fig.~\ref{fig:falarm} for $N=7$. Finally, in this figure, it must be noted that as $P_f$ approaches $\alpha = 0.1$, the advantage for collaborative spectrum sensing diminishes drastically as the network converges towards the non-cooperative case.
\begin{figure}[t]
\begin{center}
\includegraphics[angle=0,width=\columnwidth]{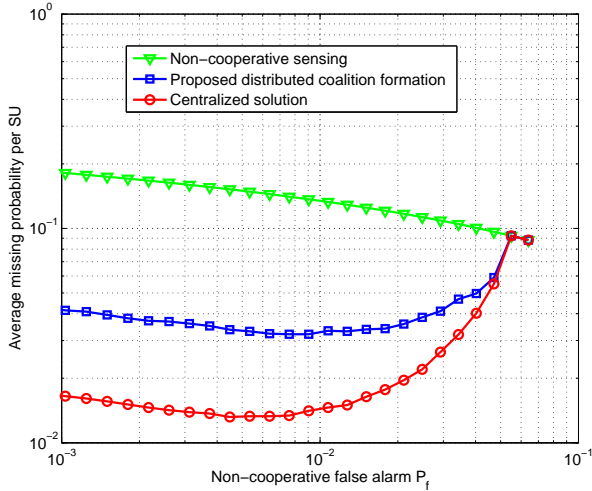}
\end{center}\vspace{-0.5cm}
\caption {Average missing probabilities per SU vs. non-cooperative false alarm $P_f$ (or energy detection threshold $\lambda$) for $N=7$~SUs.} \label{fig:perffal}\vspace{-0.7cm}
\end{figure}

In Fig.~\ref{fig:snapshot}, we show a snapshot of the network structure resulting from the proposed distributed algorithm (dashed line) as well as the centralized approach (solid line) for $N=7$ randomly placed SUs and a non-cooperative false alarm $P_f = 0.01$. We notice that the structures resulting from both approaches are almost comparable, with nearby SUs forming collaborative coalitions for improving their missing probabilities. However, for the distributed solution, SU~4 is part of coalition $S_1=\{1,2,4,6\}$ while for the centralized approach SU~4 is part of coalition $\{3,4,5\}$. This difference in the network structure is due to the fact that, in the distributed case, SU~4 acts selfishly while aiming at improving its own utility. In fact, by merging with $\{3,5\}$ SU~4 achieves a utility of $\phi_4(\{3,5\})=0.9859$ with a missing probability of $0.0024$ whereas by merging with $\{1,2,6\}$ SU~4 achieves a utility of $\phi_4(\{1,2,4,6\})=0.9957$ with a missing probability of $0.00099$. Thus, when acting autonomously in a distributed manner, SU~4 prefers to merge with $\{1,2,6\}$ rather than with $\{3,5\}$ regardless of the optimal structure for the network as a whole. In brief, Fig.~\ref{fig:snapshot} shows how the cognitive network structures itself for both centralized and distributed approaches.
\begin{figure}[t]
\begin{center}
\includegraphics[angle=0,width=\columnwidth]{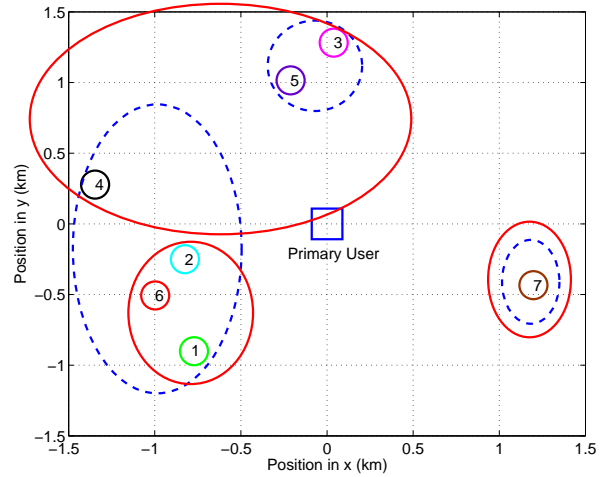}
\end{center}\vspace{-0.5cm}
\caption {Final coalition structure from both distributed (dashed line) and centralized (solid line) collaborative spectrum sensing for $N=7$ SUs.} \label{fig:snapshot}\vspace{-0.7cm}
\end{figure}

Furthermore, in Fig.~\ref{fig:mob} we show how our distributed algorithm in Table~\ref{tab:alg1} handles mobility during Phase~2 (adaptive coalition formation). For this purpose, after the network structure in Fig.~\ref{fig:snapshot} has formed, we allow SU~1 to move horizontally along the positive x-axis while other SUs are immobile. In Fig.~\ref{fig:mob}, at the beginning, the utilities of SUs $\{1,2,4,6\}$ are similar since they belong to the same coalition. These utilities decrease as SU~1 distances itself from $\{2,4,6\}$. After moving $0.8$~km SUs $\{1,6\}$ split from coalition $\{1,2,4,6\}$ by Pareto order as $\phi_1(\{1,6\})=0.9906 > \phi_1(\{1,2,4,6\})=0.99$, $\phi_6(\{1,6\})=0.9906 > \phi_6(\{1,2,4,6\})=0.99$, $\phi_2(\{2,4\})=0.991 > \phi_2(\{1,2,4,6\})=0.99$ and $\phi_4(\{2,4\})=0.991 > \phi_4(\{1,2,4,6\}=0.99)$ (this small advantage from splitting increases as SU~1 moves further). As SU~1 distances itself further from the PU, its utility and that of its partner SU~6 decrease. Subsequently, as SU~1 moves $1.4$~km it finds it beneficial to split from $\{1,6\}$ and merge with SU~7. Through this merge, SU~1 and SU~7 improve their utilities. Meanwhile, SU~6 rejoins SUs $\{2,4\}$ forming a 3-SU coalition $\{2,4,6\}$ while increasing the utilities of all three SUs. In a nutshell, this figure illustrates how adaptive coalition formation through merge and split operates in a mobile cognitive radio network. Similar results can be seen whenever all SUs are mobile or even the PU is mobile but they are omitted for space limitation.\\
\begin{figure}[t]
\begin{center}
\includegraphics[angle=0,width=\columnwidth]{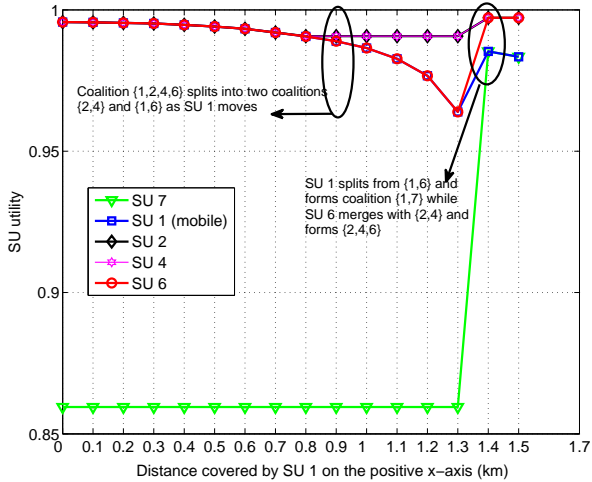}
\end{center}\vspace{-0.5cm}
\caption {Self-adaptation of the network's topology to mobility through merge-and-split as SU~1 moves horizontally on the positive x-axis.} \label{fig:mob}\vspace{-0.6cm}
\end{figure}
In Fig.~\ref{fig:numus}, for a network of $N=30$ SUs, we evaluate the sizes of the coalitions resulting from our distributed algorithm and compare them with the the upper bound $M_{\max}$ derived in Theorem~1. First and foremost, as the non-cooperative $P_f$ increases, both the maximum and the average size of the formed coalitions decrease converging towards the non-cooperative case as $P_f$ reaches the constraint $\alpha = 0.1$. Through this result, we can clearly see the limitations that the detection-false alarm probabilities trade off for collaborative sensing imposes on the coalition size and network topology. Also, in Fig.~\ref{fig:numus}, we show that, albeit the upper bound on coalition size $M_{\max}$ increases drastically as $P_f$ becomes smaller, the average maximum coalition size achieved by the proposed algorithm does not exceed $4$ SUs per coalition for the given network with $N=30$. This result shows that, in general, the network topology is composed of a large number of small coalitions rather than a small number of large coalitions, even when $P_f$ is small and the collaboration possibilities are high.\vspace{-0.1cm}

\section{Conclusions}\label{sec:conc}
In this paper, we proposed a novel distributed algorithm for collaborative spectrum sensing in cognitive radio networks. We modeled the collaborative sensing problem as a coalitional game with non-transferable utility and we derived a distributed algorithm for coalition formation. The proposed coalition formation algorithm is based on two simple rules of merge-and-split that enable SUs in a cognitive network to cooperate for improving their detection probability while taking into account the cost in terms of false alarm probability. We characterized the network structure resulting from the proposed algorithm, studied its stability and showed that a maximum number of SUs per coalition exists for the proposed utility model. Simulation results showed that the proposed distributed algorithm reduces the average missing probability per SU up to $86.6\%$ compared to the non-cooperative case. The results also showed how, through the proposed algorithm, the SUs can autonomously adapt the network structure to environmental changes such as mobility. Through simulations, we also compared the performance of the proposed algorithm with that of an optimal centralized solution.
\begin{figure}[!t]
\begin{center}
\includegraphics[angle=0,width=\columnwidth]{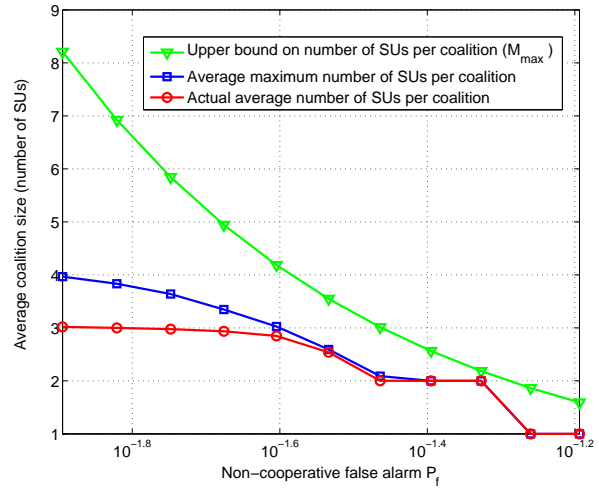}
\end{center}\vspace{-0.5cm}
\caption {Maximum and average coalition size vs. non-cooperative false alarm $P_f$ (or energy detection threshold $\lambda$) for the distributed solution for $N=30$~SUs.} \label{fig:numus}\vspace{-0.7cm}
\end{figure}
\vspace*{-0.6em}
\def\baselinestretch{0.95}
\bibliographystyle{IEEEtran}
\bibliography{references}

\end{document}